%% file: main.tex
\documentclass[a4paper]{article}
\usepackage{graphicx}
\usepackage{amsmath, amsthm}
\usepackage{algorithm}
\usepackage{algpseudocode}
\usepackage{float}
\usepackage{soul}
\usepackage{caption}
\usepackage{subcaption}
\usepackage{wrapfig}
\usepackage{hyperref}

\newtheorem{theorem}{Theorem}[section]
\newtheorem{lemma}[theorem]{Lemma}

\newtheorem{proposition}[theorem]{Proposition}

\newtheorem{definition}[theorem]{Definition}

\title{Cooperative Guarding in Polygons with Holes}

\providecommand{\keywords}[1]
{
  \textbf{Keywords:} #1
}

\date{}

\title{Guarding Polygons with Holes}
\author{John Augustine$^{1}$, Srikkanth Ramachandran$^{1}$  \\
        \small $^{1}$Indian Institute of Technology Madras \\
}

\begin{document}

\maketitle

\begin{abstract}
    We study the Cooperative Guarding problem for polygons with holes in a mobile multi-agents setting. Given a set of agents, initially deployed at a point in a polygon with $n$ vertices and $h$ holes, we require the agents to collaboratively explore and position themselves in such a way that every point in the polygon is visible to at least one agent and that the set of agents are visibly connected. We study the problem under two models of computation, one in which the agents can compute exact distances and angles between two points in its visibility, and one in which agents can only compare distances and angles.
    In the stronger model, we provide a deterministic $O(n)$ round algorithm to compute such a cooperative guard set while not requiring more than $\frac{n + h}{2}$ agents and $O(\log n)$ bits of persistent memory per agent. In the weaker model, we provide an $O(n^4)$ round algorithm, that does not require more than $\frac{n+2h}{2}$ agents.
\end{abstract}


\input{intro.tex}

\input{warmup.tex}

\input{model.tex}

\input{algorithm.tex}

\input{conclusion.tex}

\bibliography{references}

\end{document}

%% file: intro.tex

\section{Introduction}

The {\it Art Gallery Problem} is a classical computational geometry problem which asks for the minimum number of guards required to completely guard the interior of a given art gallery. 
The art gallery is modelled as a simple polygon and guards are points on or inside the polygon. 
A set of guards is said to guard the art gallery, if every point in the art gallery is visible to at least one guard. 
This problem was first posed by Klee in 1973 and since then the problem has been of interest to researchers. 
The problem and its many variations have been well-studied over the years. Chv\'{a}tal \cite{CHVATAL197539} was the first to show that $\lfloor \frac{n}{3} \rfloor$ guards are sufficient and sometimes necessary to guard a polygon with $n$ vertices. 
Fisk \cite{FISK1978374} proved the same result via an elegant coloring argument, which also leads to an $O(n)$ algorithm that computes such a guard set. 
The problem of finding the minimum number of guards was shown to be NP-hard (Lee and Lin \cite{1057165}) as well as APX-hard(Eidenbenz {\it et al.}\cite{Eidenbenz2001}). 
Consequently, researchers have focussed on approximation algorithms (Ghosh~\cite{G87}, Deshpande {\it et al.} \cite{deshpande2007pseudopolynomial},  and Bhattacharya {\it et al.} ~\cite{BGP17}).


In this paper, we study a variant of the classical Art Gallery Problem known as the {\em Cooperative Guards} problem in polygons with holes from a distributed multi-agents perspective. 
The {\em Cooperative Guards} problem is similar to the {\it Art Gallery Problem} with the additional constraint that the visibility graph of the guards, i.e., the graph with guards as vertices and edges between guards that can see each other, should form a single connected component. 
This implies that if the guards can communicate through line of sight, then any two guards can communicate with each other (either directly or indirectly through intermediate guards). 
The {\it Cooperative Guards} problem was first introduced by Liaw {\it et al.} in 1993 ~\cite{Liaw1993TheMC}, in which they show that this variant of the {\it Art Gallery Problem} is also NP-hard.

The pursuit of solving the {\it Cooperative Guards} problem for us is primarily motivated by the search for distributed multi-agent exploration algorithms for agents deployed in the polygon. 
One of the easiest ways to explore the polygon is by maintaining connectivity through line of sight and having the agents cooperatively send exploratory agents and expand the area of the polygon that they can collectively see. Such an algorithm would require the agents to be visibly connected and if at the end they are collectively able to see the entire polygon, then they form a solution to the {\it Cooperative Guards} problem. 
This lead us to first tackle the {\it Cooperative Guards} problem in the centralized setting and adapt it to the distributed multi-agent model. 
Note that once the agents are positioned to cooperatively guard the polygon, they can collectively solve \textit{any} computational geometry problem on the polygon by executing a distributed algorithm on their visibility graph. 
Hence our algorithm can also serve as an initial pre-processing step for distributed multi-agents to solve other problems on the polygon, for example they can compute the number of vertices of the polygon, diameter of the polygon etc. 
Our focus throughout this paper is to find time optimal algorithms for computing a set of cooperative guards on the polygon (not necessarily minimum). 
Analogous to Fisk's \cite{FISK1978374} coloring argument for the classical {\it Art Gallery Problem}, we provide an efficient algorithm that computes a set of visibly connected guards using no more than $\frac{n + h - 2}{2}$ guards. 

\subsection{Our Contributions}

Our goal is to design efficient deterministic algorithms for placing  cooperative guards in polygons with holes. While we warmup with a centralized algorithm, our main focus is on constructing distributed algorithms executed by mobile agents.

In Section~\ref{warmup} we first present a centralized algorithm for the Cooperative Guards problem in polygonal region with $n$ vertices and $h$ holes, that does not require more than $\frac{n+2h-2}{2}$ guards.
The centralized algorithm presented here runs in time $O(n+T(n, h))$ where $T(n, h)$ is the time required to triangulate a polygon with $n$ vertices and $h$ holes, i.e. after the initial step of triangulating the polygon, the algorithm takes linear time and places no more than $\frac{n + 2h - 2}{2}$ guards. 
We then use an observation from Zylinski \cite{zylinski2005cooperative} to reduce the number of guards to $\frac{n+h-2}{2}$. The reduction is slower and runs in $O(n^2)$ time.

In Section~\ref{model} we propose two distributed models of computation, one in which the agents can perceive exact distances (which we call {\em depth perception}) and hence can compute co-ordinates (with respect to some common reference frame) to map out the polygon, and another in which the agents receive only a combinatorial view of their visibility.  

In Section~\ref{sec:algorithm} we present a distributed algorithm (in the stronger {\em depth perception} model) that runs in $O(n)$ rounds and does not require more than $\lfloor \frac{n + h - 2}{2} \rfloor$ agents. Our distributed algorithm emulates the centralized algorithms presented in Section~\ref{warmup} except that the polygon is not known in advance. We simultaneously explore and incrementally construct a  triangulation. 

In Section~\ref{sec:algorithm:proximity} we present a distributed algorithm for the weaker {\em proximity perception } model, which takes $O(n^4)$ rounds, and requires $\frac{n+2h-2}{2}$ agents. 

A summary of the results is presented in Table \ref{tab:results}.



\begin{table}[ht]
    \centering
    \begin{tabular}{|c|c|c|c|c|}
        \hline
        Model \& Algorithm & Rounds & Broadcasts & Persistent & Guards  \\
         &  &  & Memory &   \\
        \hline
        Depth Perception & $O(n)$ & $O(n)$ & $O(n \log n)$ & $(n+h-2)/2$\\
        (Large memory Thm \ref{thm:dist:alg1}) & & & & \\
        \hline
        Depth Perception & $O(n)$ & $O(n \log n)$ & $O(\log n)$ & $(n+2h-2)/2$\\
        (Small memory Thm \ref{thm:dist:alg3}) & & & & \\
        \hline
        Depth Perception & $O(n)$ & $O(n (h + \log n))$& $O(\log n)$ & $(n+h-2)/2$\\
        (Improving guards Thm \ref{thm:dist:alg4}) & & & & \\
        \hline 
        Proximity Perception (Thm \ref{thm:dist:proximity}) & $O(n^4)$ & $O(n^4)$& $O(\log n)$ & $(n+2h-2)/2$\\ 
        \hline
    \end{tabular}
    \caption{A summary of our results with respective theorem numbers indicated in the first column.}
    \label{tab:results}
\end{table}


\subsection{Related Work}
Most of the algorithms developed for solving the {\it Art Gallery Problem} and its variants triangulate the polygon in the first step. 
So we first look at the best known algorithms for triangulation.
Chazelle \cite{Chazelle1991} provided a linear time algorithm for computing the triangulation of a simple polygon with $n$ vertices, however the algorithm is known to be quite complicated. 
Simpler $O(n \log n)$ algorithms were established long before by Gary {\it et al.} \cite{garey1978triangulating}. For polygons with holes the algorithm by Bar-Yehuda and Chazelle \cite{BARYEHUDA1994} can be used to obtain an $O(n + h \log^{1 + \epsilon} h)$ time algorithm. 
For polygons with holes the triangulation problem has an $\Omega(n \log n)$ lower bound (as shown by Asano {\it et al.}~\cite{ASANO1986221}). 
Our algorithm involves triangulation as the first (and most time consuming) step and thus the run time of all algorithms is at best $O(n \log n)$ and slightly better when $h = o(n)$ if Bar-Yehuda and Chazelle's\cite{BARYEHUDA1994} algorithm is used. 

Chv\'{a}tal \cite{CHVATAL197539} and Fisk \cite{FISK1978374} independently proved that $\lfloor \frac{n}{3} \rfloor$ guards are always sufficient and occasionally necessary to guard a simple polygon with $n$ vertices, with the latter providing an elegant constructive proof which can construct the desired set of guards in linear time.

The classical {\it Art Gallery Problem} for polygons with holes was first solved by Souvaine {\it et al.} \cite{BjorlingSachs1995}. 
They show that $\frac{n + h}{3}$ guards are sufficient (by presenting an $O(n^2)$ time algorithm) and sometimes necessary to guard a polygonal region with $n$ vertices and $h$ holes. 

The {\it cooperative guards} problem was first posed by Liaw \cite{Liaw1993TheMC} where they provide a solution for $k-$spiral polygons and also show that the problem is NP-hard. 
Ahlfeld and Hecker \cite{AhlfeldUweandHecker92} were the first to prove that for simple polygons $\lfloor \frac{n-2}{2} \rfloor$ guards are sufficient and sometimes necessary. 
Their algorithm relies on analysing the length of the longest degree $2$ chains in the dual graph. The dual graph is a graph corresponding to a triangulation of the polygon. 
The vertices of the polygon correspond to a triangle in the triangulation and an edge is present between two vertices if their corresponding triangles share a common edge. 
Hernandez-Penalver \cite{hernandez1994controlling} proved the same result by induction on the dual graph.
Pinciu \cite{Pinciu} also provided an alternate proof of the same result by extending Fisk's coloring argument. Zylinski \cite{zylinski2007cooperative} also provided an alternate soluton based on the diagonal graph. 
The diagonal graph is the graph induced by the diagonals of the triangulation. 
A vertex cover of the diagonal graph gives a cooperative guard set with no more than $\lfloor \frac{n}{2} \rfloor$ guards. 

For polygons with holes, Ahlfeld and Hecker \cite{AhlfeldUweandHecker92} were the first to show that $\lfloor \frac{n + 2h - 2}{2} \rfloor$ guards are sufficient, however the bound is not tight. 
Zylinski \cite{zylinski2005cooperative} showed that for polygons with one hole, $\lfloor \frac{n-1}{2} \rfloor$ guards are sufficient and sometimes necessary.
Zylinski \cite{zylinski2007cooperative} provided a better upper bound of $\frac{n + h - 1}{2}$ guards using ideas established by Sachs and Souvaine \cite{BjorlingSachs1995} that leads to an $O(n^2)$ algorithm. 
Zylinksi \cite{zylinski2005cooperative} also provides a different linear time algorithm that takes no more than $\lfloor \frac{n + 2h - 2}{2} \rfloor$ guards.
Zylinksi's book \cite{zylinski2007cooperative} contains a more extensive survey of these algorithms.

The {\it Cooperative Guards} problem was first studied in the mobile-agents setting by Obermeyer {\it et al.} \cite{OGB11}. 
Ashok {\it et al.}\cite{Bharat2020} provides a better $O(\min(D^2, n))$ time algorithm where $D$ is the maximum hop distance between any two points in the polygon, however it does not extend to polygons with holes and requires upto $n$ agents for exploration.

In the mobile agents setting, several models have been proposed and studied over the years. A survey of such models is available in~\cite{DCME_book}. In most of them, the agents operate either on a graph or on the 2D-plane without obstacles or bounding structures. 

In all the models studied in this paper, the movement of the agents is on a region of the 2D-plane that is bounded by a polygon (possibly with holes) that is not known in advance.  Consequently we need to learn and contend with the geometric complexities of the polygon to solve problems in our model(s).  There have been a few papers on mobile agents operating within bounded simple polygons (Ashok {\em et al.}~\cite{Bharat2020}). To the best of our knowledge, the only other work on mobile agents operating on a polygon with holes is by Obermeyer {\em et al.}~\cite{OGB11}. They provide an algorithm that takes no requires no more than $\lfloor \frac{n + 2h - 1}{2} \rfloor$ agents to collectively explore the polygon and deploy the agents so that they cooperatively guard the polygon. The running time of their algorithm is $O(n^2)$ rounds in the worst case.

%% file: warmup.tex

\section{Preliminaries}\label{warmup}

\subsection{Polygons and Visibility}

\begin{figure}
\includegraphics[width=\textwidth]{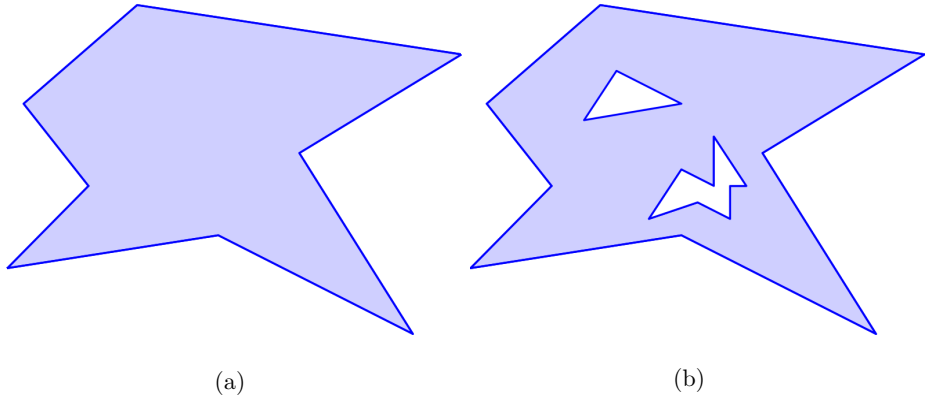}
\caption{(a) A simple polygon with 8 vertices, (b) A polygon with 19 vertices and 2 holes} \label{fig:poly_ex}
\end{figure}

\begin{definition}\label{def:polygon_simple}
Simple Polygon. A simple polygon is described by its boundary $B$ which can be represented by a sequence of alternating vertices and sides $v_1, e_1, v_2, e_2, \dots v_{n-1}, e_{n-1}, v_n, e_n, v_1$ such that (i) the side $e_i$ is the line segment joining $v_{i}$ and $v_{i \mod n + 1}$, (ii) all vertices $v_1, v_2, \dots v_n$ are distinct and (iii) no two distinct sides properly intersect. The boundary $B$ forms a closed loop in the plane and the polygon $P$, is the set of points that are either on or in the interior of $B$. The exterior of the polygon $P$, is the set of points not belonging to $P$.
\end{definition}

\begin{definition}\label{def:polygon_holes}
Polygon with holes. A polygon $P$ with $n$ vertices and $h$ holes is described by a simple polygon $D$ and a set of $h$ mutually disjoint simple polygons (called holes) all lying in the interior of $D$. The polygon $P$ is the set of points that do not lie in the exterior of $D$ and those that do not lie on the interior of the $h$ holes. The total number of vertices of $P$, given by $n$, is the sum of the number of vertices of all the simple polygons (including $D$ and the $h$ mutually disjoint polygons) describing $P$.
\end{definition}
By definition all simple polygons have at least $3$ vertices and for all polygons with $h$ holes, $n \geq 3h + 3$ (the boundary and all holes must have $3$ vertices each). 
\begin{definition}\label{def:visbility}
Visibility of points. Two points $p, q \in P$ are said to be visible to each other in a polygon $P$, if the line segment joining them does not intersect the exterior of $P$. 
\end{definition}
Note that we allow points to lie on the line segment. In particular point objects in $P$ do not obstruct visibility. 
\begin{definition}\label{def:visbility_polygon}
Visibility Polygon. For any point $p \in P$, the set of points $q$ such that $p, q$ are visible in $P$ is the visibility polygon of $p$ and is denoted by $VP(p)$.
\end{definition}

\begin{definition}\label{def:visbility_vertices}
Vertex-limited Visibility Polygon. The vertex limited visibility polygon of a point $p$ in a polygon $P$ is a sub-polygon of $VP(p)$, with maximum area, all of whose vertices are vertices of $P$. 
\end{definition}

\begin{figure}
\includegraphics[width=\textwidth]{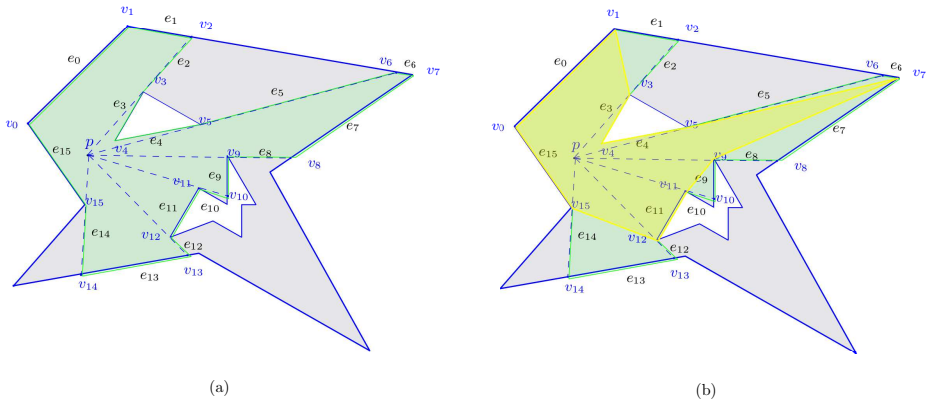}
\caption{(a) Visibility polygon of a point p (shaded green), (b) Vertex limited visibility polygon (shaded yellow)} \label{fig:visibility_ex}
\end{figure}
The vertex limited visibility polygon can be obtained by choosing any vertex, $v$ in $VP(p)$, that does not belong to $P$, removing $v$ and adding a side between the vertices adjacent to $v$ in the boundary of $VP(p)$. If the procedure is repeated until all vertices of the sub-polygon obtained is a vertex of $P$, we obtain the vertex limited visibility polygon of $p$.
Figure \ref{fig:visibility_ex} (a) contains an example of a visibility polygon of a point $p$ and the corresponding vertex limited visibility polygon.

\begin{definition}\label{def:crop}
Crop(VP, s, p). We define an operation $crop$ that takes as input a visibility polygon $VP$, a diagonal $s$ of the visibility polygon and a point $p$ that is on or in the interior of the $VP$, but NOT on the side $s$. The polygon $crop(VP, s, p)$ is the portion of $VP$ that lies on or in the opposite of the diagonal $s$ that contains $p$. The operation is useful when we look at exploration algorithms on the polygon. Intuitively the operation $crop(VP, s, p)$ denotes some region of the polygon that is yet to be explored. An illustration is shown in Figure \ref{fig:look_cycle} (b).
\end{definition}

\begin{definition}\label{def:triangulation}
Triangulation and Dual graph. For a polygon $P$ with $n$ vertices and $h$ holes, a triangulation is a partition of the polygon $P$ into triangles such that (i) all vertices of all triangles are vertices of $P$ and (ii) the interiors of the triangles are mutually disjoint. The corresponding dual graph is a simple graph whose nodes represent the triangles of the triangulation and there exists an edge between two nodes in the graph iff their corresponding triangles share an edge.
\end{definition}

It is known that for any polygon $P$ with $n$ vertices and $h$ holes, and any valid triangulation of $P$, the total number of nodes in the dual graph is $n + 2h - 2$ and the total number of edges is $n + 3h - 3$. This can be proved using induction. Note that the dual graph is a sparse (and more specifically also a planar graph) with only $h-1 = O(n)$ extra edges. Additionally when $h = 0$, the dual graph is a tree. A simple method (although not computationally efficient) to construct a valid triangulation is to construct any maximal set of non-intersecting diagonals of the polygon. These diagonals create the desired triangulation. Note that any diagonal splits the polygon into two sub-polygons and we can recursively construct the triangulation for each piece. This argument can be used to build an inductive proof of the statements made above and we will use these simple concepts to build distributed algorithms. 


In this section we first present a simple centralized algorithm that given a polygon $P$ with $n(\geq 4)$ vertices and $h$ holes, computes a cooperative guard set with no more than $\frac{n + 2h - 2}{2}$  guards. 
We then show how to update the co-operative guard set, when the polygon is slightly modified. The update method is more efficient than computing the guard set from scratch and will be emulated by the distributed algorithms in Section \ref{sec:algorithm:proximity}. We will then briefly discuss reducing the number of guards to $\frac{n+h-2}{2}$ using a technique given by Sachs and Souvaine \cite{BjorlingSachs1995}. 

\subsection{A Centralized Algorithm for Cooperative Guarding}
The idea presented here is exactly the same as that presented by Ashok {\it et al.}~\cite{Bharat2020}. We write the same procedure in a depth-first search manner and extend it for polygons with holes. 
In the first step we construct a triangulation, $T$ (See Definition \ref{def:triangulation}), of the polygon and the corresponding dual graph $\mathcal{D}_T$. We call a set of $3$ nodes in the graph $\mathcal{D}_T$ that lie on a path of length $2$ as a ``triplet". From the pigeon-hole principle, all triangles corresponding to each triplet share a common vertex. We place the guards at this common vertex. If two triplets share a common node $c$, then their corresponding common vertices will lie on the triangle $c$ and are therefore visibly connected. Using the above ideas, we can construct a set of $\frac{n+2h-2}{2}$ guards, each of which is placed on a vertex of the polygon $P$. Algorithm \ref{alg:dfs} is a depth first search procedure that computes such a guard set. First construct a spanning tree rooted at any node with degree at most $2$, which is a binary tree. We construct the triplets from the binary tree in a bottom-up fashion (i.e. triplets are formed with nodes with higher depth first). At the bottom-most level we will encounter triplets of two types, as shown in Figure~\ref{fig:triplets}. Once the triplets is placed according to the cases in Figure~\ref{fig:triplets}, we can remove the two nodes at higher depths, and retain only the node with the lowest depth. Repeating the above procedure yields the required set of triplets. In the depth-first search procedure, instead of explicitly deleting the nodes, we mark them as $\mathsf{covered}$. 

\begin{figure}
\begin{center}
\includegraphics[width=0.2\textwidth]{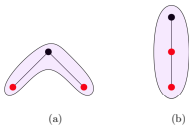}
\end{center}
\caption{(a) Node has two children, (b) Node has one child. Nodes in red are marked $\mathsf{covered}$ after forming the triplets (shaded purple)} \label{fig:triplets}
\end{figure}

\begin{algorithm}[ht] 
    \caption{Centralized algorithm for the cooperative guards problem for polygons with holes.} 
    \label{alg:central} 
    \begin{algorithmic}[1]
        \State Compute a triangulation $T$ of the polygon $P$ and the corresponding weak dual graph $\mathcal{D}_T$
        \State $t \gets $ An array of size $|V(\mathcal{D}_T)|$ each of which can store a set of size 3
        \State $v$ $\gets$ any node in $\mathcal{D}_T$ whose degree is not $3$ 
        \State $\mathcal{G}_T \gets $ Any spanning tree of $\mathcal{D}_T$ rooted at $v$ \Comment{Always a binary tree}
        \State \Comment{The rooted spanning tree also includes the parent $\texttt{p}(u)$ for each node $u$. By convention, for the root $v$, $\texttt{p}(v) = v$}
        \State \Call{DFS}{$v$} \Comment{Perform depth-first search as described in Algorithm \ref{alg:dfs}}
        \ForAll {triplets computed by \Call{DFS}{$v$}} \Comment{Choose the non-empty triplets in the array $t$}
            \State Place a guard at a common vertex of all triangles in the triplet
        \EndFor
    \end{algorithmic}
\end{algorithm}

\begin{algorithm}[ht] 
    \caption{Depth first search procedure for computing a connected set of triplets.}
    \label{alg:dfs} 
    \begin{algorithmic}[1]
        \Procedure {DFS}{$u$} 
            \State Mark $u$ as $\textsf{uncovered}$
            \ForAll {children $v$ of $u$ in $\mathcal{G}_T$} 
                \State \Call{DFS}{$v$}
            \EndFor 
            \State $C \gets $ Set of all children of $u$ that are $\mathsf{uncovered}$
            \State $T \gets $ Set of the three lowest depth nodes in $C \cup \{u, \mathsf{p}(u), \mathsf{p}(\mathsf{p}(u))\}$ \Comment{Refer Figure \ref{fig:triplets}}
            \State Mark the two lowest depth nodes of $T$ as $\mathsf{covered}$
        \EndProcedure
    \end{algorithmic}
\end{algorithm}

\begin{figure}
\includegraphics[width=0.4\textwidth]{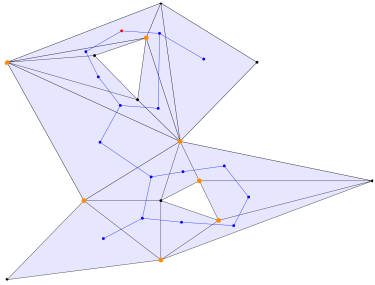}
\caption{An example polygon with $n=14$ vertices and $h=2$ holes. The diagonals form the triangulation $T$. The blue (and red) nodes and edges represent the weak dual graph $\mathcal{D}_T$ and are embedded in the polygon only for clarity. The orange nodes is the set of vertex guards computed by Algorithm \ref{alg:central} when DFS starts from the red node.} \label{fig:poly_example}
\end{figure}

\begin{figure}
\includegraphics[width=0.5\textwidth]{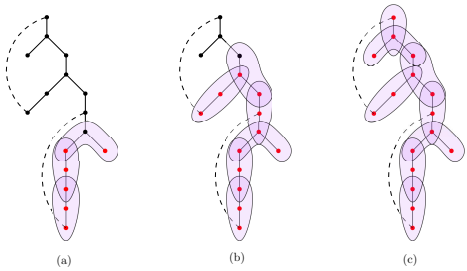}
\caption{Three snapshots of the triplet formation by Algorithm \ref{alg:central}. The dual graph is the same as that in Figure \ref{fig:poly_example}. The dashed lines denote the back edges during DFS traversal and the red nodes denote $\mathsf{covered}$ nodes. The purple regions represent the triplets formed.} \label{fig:dual_graph}
\end{figure}

\begin{theorem}\label{thm:alg2}
    For all polygons with $n \geq 4$, Algorithm \ref{alg:central} constructs a set of co-operative guards of size at most $\frac{n + 2h - 2}{2}$ in time $O(n \log n)$.
\end{theorem}

\begin{proof}
    The first step in Algorithm \ref{alg:central} is the most time consuming step. Many algorithms have been developed for triangulating polygons with holes that run in time $O(n \log n)$, we can use any of them. We then choose any node in the weak dual graph that has degree at most $2$. Such a node is guaranteed to exist because the number of vertices is $|V(\mathcal{D})| = n + 2h - 2$ and the sum of degrees of all nodes is $2 |E(\mathcal{D}_T)| = 2n + 6h - 6 < 3|V(\mathcal{D}_T)|$. Let $\mathcal{G}_T$ be the DFS tree obtained from line 4. We will show that the following invariant holds by induction on the number of triplets that are formed.
    \begin{proposition}
        During the course of the Algorithm \ref{alg:dfs}, when $t$ triplets are formed, the total number of nodes that are $\mathsf{covered}$ is at least $2t$.
    \end{proposition}
    \begin{proof}
        When $t = 0$ triplets are formed the proposition is trivial. Suppose $t(\geq 0)$ triplets have been constructed and the number of $\mathsf{covered}$ nodes in the graph is at most $2t$, then the next triplet is formed by either lines 12 or 15 of Algorithm \ref{alg:dfs}. By definition, the nodes in $C$ are previously $\mathsf{uncovered}$ before the triplet is formed. We also observe that the DFS traversal never visits those nodes that are marked $\mathsf{covered}$ again, therefore node $u$ must be $\mathsf{uncovered}$ in line 12. Therefore, after the new triplet is formed, the number of $\mathsf{uncovered}$ nodes is at least $2t + 2$, which proves the induction step. We conclude, by induction, that the proposition is true.
    \end{proof}
    At the end of the execution, let $t$ be the total number of triplets placed. The total number of $\mathsf{covered}$ nodes is at most $n$, therefore,
    \begin{equation}
        2t \leq |V(\mathcal{D})| \implies t \leq \frac{n+2h-2}{2}
    \end{equation}
    
    We now show that the guards are visibly connected and guard the entire polygon. For any node $u \in V(\mathcal{D}_T)$ let $t_u$ be the triplet computed during the DFS call at node $u$. If no triplet is formed then $t_u = \emptyset$. Define $S_u := \{t_v : v \text{ belongs to the sub-tree of } u\}$. If $v$ is the root of the DFS tree, then $S_v$ is just the set of all triplets computed. The height of a node is the maximum distance of the node to any leaf in its sub-tree.
    \begin{proposition}\label{prp:guard_set}
        For every non-leaf node $u \in \mathcal{G}_T$,
        \begin{enumerate}
            \item the set of guards formed by $S_u$ are visibly connected.
            \item every node in the sub-tree of $u$ belongs to at least one triplet of $S_u$.
        \end{enumerate}
    \end{proposition}
    \begin{proof}
        We use induction on the height of the nodes $u$. If the node $u$ has height $1$, then there are only two possibilities and they are shown in Figure \ref{fig:triplets}. We verify, by inspection, that the proposition is true for these cases. To ease notation, we say that a set of triplets is connected if the corresponding guards computed are visibly connected.
        
        Suppose all nodes at height $h (\geq 1)$ satisfy the proposition. Consider any node $u$ at height $h + 1$. 
        
        If any of child of $u$ is $\mathsf{uncovered}$ then $t_u$ contains both $u$ and the child. If $c$ is a child of $u$ that is $\mathsf{covered}$, then $t_c$ contains both $c$ and $u$. In both cases we see that all children of $u$ and $u$ are present in one of triplets of $S_u$. All other nodes are present in $S_u$ by the induction hypothesis. Therefore the second statement of the proposition is true for node $u$.
        
        Similarly we can show that $S_u$ is connected. If all children of $u$ are $\mathsf{covered}$, then no triplet is formed during DFS call at $u$. If $u$ contains only one child, $c$, then $S_u = S_c$ and is connected by the induction hypothesis. If $u$ contains two children, $c_1, c_2$, then $t_{c_1}$ and $t_{c_2}$ both contain $u$ and hence $S_u = S_{c_1} \cup S_{c_2}$ is connected. If $u$ contains a child $c_1$ that is $\mathsf{uncovered}$ then $t_u$ contains both $c_1$ and $u$ and so $S_{c_1} \cup \{t_u\}$ is connected. If $u$ contains only a single child, we are done. If $u$ contains another child, $c_2$, then either $t_u$ contains $c_2$ ($c_2$ is $\mathsf{uncovered}$) or $S_{c_2}$ contains a triplet containing $u$ ($c_2$ is $\mathsf{uncovered}$). In both cases $S_u$ is connected.
        
        From the above arguments, we can conclude that for all nodes $u$ at height $h+1$, the proposition is true and so by induction, we conclude that the proposition is true for all nodes $u \in V(\mathcal{G}_T)$.
    \end{proof}
    Applying the Proposition \ref{prp:guard_set} for the root node, we get that every triangle contains a guard and that the set of guards computed are visibly connected and thus form a solution to the {\em Cooperative Guards} problem. 
    All the steps after the triangulation are linear and so the total time complexity is $O(n \log n)$ which completes the proof of Theorem \ref{thm:alg2}.
\end{proof}

\subsection{Update Algorithm} \label{alg:central_update}
Algorithm \ref{alg:central} computes the guard set by constructing the dual graph of the polygon, by computing a decomposition of the dual graph into connected triplets. The set of triplets can then be used to compute the cooperative guard set. In this subsection, we look at efficient ways to update the set of triplets when the dual graph is modified by the addition or deletion of a leaf.

We observe that for any node $u$, the set of triplets computed during DFS calls of nodes that are neither ancestors nor descendents are independent. Therefore while adding a leaf to a node $v$, only the triplets computed during DFS calls of $v$ and its ancestors need to be changed. By looking at the triangle placed during DFS call at a child of $v$, we can determine whether the child was $\mathsf{covered}$ during DFS call at node $v$. This information is sufficient to update $t[v]$ by executing lines 11-19 in Algorithm \ref{code:central_update}. Repeating the same for the parents of $v$, we can update the set of triplets in $O(d(u))$ time where $d(u)$ is the depth of the node $u$. 
The details are presented in Algorithm \ref{code:central_update}.
\begin{algorithm}[htbp] 
    \caption{Algorithms for updating the desired triplets when $\mathcal{G}_T$ is changed by the addition of a leaf, note that the update is to be performed such that $\mathcal{G}_T$ remains a binary tree.} 
    \label{code:central_update} 
    \begin{algorithmic}[1]
        \Procedure{findTriplet}{$v$}
            \State C $\gets$ Set of children, $c$, of $v$ such that $v \not \in t[c]$
            \If{$|C|$ equals $1$}
                \State \Return $C \cup \{v, \texttt{p}(v)\}$ \Comment{Form a triplet with the node in $C$, $v$ and parent of $v$} 
            \ElsIf{$|C|$ equals $2$}
                \State \Return $C \cup \{v\}$ \Comment{Form a triplet with the nodes in $C$ and $v$}
            \Else
                \State \Return $\emptyset$ \Comment{No triplet placed}
            \EndIf
        \EndProcedure
        \Procedure {addLeaf}{$v$, $l$} \Comment{Add the node $l$ as a child to the node $v$}
            \State Set $\texttt{p}(l)$ as $v$ and add $l$ as a child of $v$ \Comment{Add node $l$ to the graph}
            \State $t[l] \gets \emptyset $ \Comment{Initialise $t[l]$ to an empty set}
            \State $t[v] \gets $ \Call{findTriplet}{$v$}
            \If{$v$ is not the root}
                \State \Call{addLeaf}{$\texttt{p}(v), v$}
            \EndIf
        \EndProcedure
        \Procedure {removeLeaf}{$v$, $l$} \Comment{Remove the child $l$ of $v$}
            \State Remove $l$ as a child of $v$ \Comment{Remove node $l$ from the graph}
            \State $t[v] \gets $ \Call{findTriplet}{$v$}
            \If{$v$ is not the root}
                \State \Call{removeLeaf}{$\texttt{p}(v), v$}
            \EndIf
        \EndProcedure
    \end{algorithmic}
\end{algorithm}

\begin{theorem}
    Suppose $t$ is a set of triplets computed by Algorithm \ref{alg:dfs} for some rooted binary tree $\mathcal{G}_T$ and $\mathcal{G}_T'$ be another binary tree obtained by adding a new leaf node $l$ as a child to a node $v \in V(\mathcal{G}_T)$, then the procedure \Call{addLeaf}{$v, l$} [\ref{alg:central_update}] modifies the set of triplets $t$ to $t'$, in time linear in the depth of $v$, such that $t'$ is the set of triplets computed by Algorithm \ref{alg:dfs} when run on the graph $\mathcal{G}_T'$.
\end{theorem}

\begin{proof}
    The running time of \Call{addLeaf}{$v, l$} is linear in the depth of $v$ because in every recursive call to \Call{addLeaf}{}, the parent of $v$ is called. To see that the set of triplets is exactly the same, note that \Call{DFS}{} calls on nodes that are not ancestors of $v$ are not affected by the addition of the leaf $l$. Now to correct the position of the triplets for $v$ and its ancestors, we can simply re-execute lines 10-17 of Algorithm \ref{alg:dfs} for $v$ and its ancestors (starting from $v$ and moving up until we reach the root). In order to determine the set $C$, we had a marker for each node, $\texttt{covered}$, for each node in the dual graph, however it is sufficient to look at the triplets that were placed during DFS calls to the children. A child $c$ of a node $u$ will be $\texttt{covered}$ during an execution of \Call{DFS}{$u$} if and only if $t[c]$ computed during \Call{DFS}{$c$} contains $v$. Once the set $C$ is determined, the triplet $t[v]$ can be determined and hence the theorem is true.
\end{proof}

\begin{figure}
\includegraphics[width=0.5\textwidth]{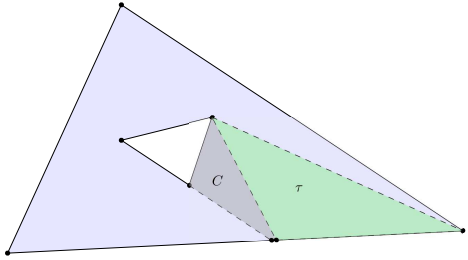}
\caption{The gray shaded region is a channel $C$, and the corresponding triangle $\tau$ is shaded green. Every vertex of $\tau$ can see the entire channel $C$, and the new polygon has 1 extra vertex (The hole vertex is deleted and two new vertices are created at the side of the triangle).} \label{fig:channel}
\end{figure}

\subsection{Reducing the number of guards}

We use the following theorem from Sachs and Souivane \cite{BjorlingSachs1995}, to reduce the Cooperative Guards problem for a polygon with $n$ vertices and $h$ holes to a simple polygon with $n + h$ vertices.

\begin{theorem}
    Given a polygon $P$ with $n$ vertices and $h$ holes, there exists a quadrilateral (called a channel) $C$ that lies on the polygon $P$ which when removed from $P$ creates a polygon $P'$ such that,
    \begin{enumerate}
        \item $P'$ has $n+1$ vertices and $h-1$ holes.
        \item There exists a triangle, $\tau$, belonging to some triangulation of $P'$, such that the channel $C$ is entirely visible from all the vertices of $\tau$.
    \end{enumerate}
\end{theorem}

The theorem is mentioned when $h=1$ as Lemma 2 in \cite{BjorlingSachs1995}, but the above theorem is discussed and proved in Section 4 of the paper where they describe the algorithm for placing guards. Naturally, we would like to repeatedly apply the above theorem $h$ times to reduce the number of holes to zero, however we need to ensure that the addition of channels does not interfere with the existence of the triangles $\tau$ for the previous channels. In Section 4 of \cite{BjorlingSachs1995}, they show that such a channel does exist, however one needs to select such a channel carefully. The computation of all channels necessary can be done in $O(n^2)$. The exact details are present in \cite{BjorlingSachs1995}. An example of such a channel is shown in Figure \ref{fig:channel}.

The reduction shows that any solution to the cooperative guards problem for the reduced simple polygon with $n + h$ vertices is also a solution to the original polygon with $n$ vertices and $h$ holes. Using Algorithm \ref{alg:central} we can obtain a cooperative guard set with no more than $\frac{n + h - 2}{2}$ guards. 

\textbf{A note on precision.} In the construction of a channel, the base of the quadrilateral $C$ (See Figure \ref{fig:channel}) has a length $\epsilon$ that can be chosen arbitrarily. As the procedure is repeated, the additional points generated as a result of the decomposition may lie inside the polygon, however the points are always intersections between (i) a line joining two points of a hole of the polygon and (ii) a line joining two points each of which might have been generated when previous holes were removed. One concern is that the precision with which the points need to compute might become exponentially large in $h$ as more points are computed. We argue that this is not the case. Consider the set of points obtained by the intersection of lines obtained by joining a pair of points of the polygon. There are $O(n^4)$ such points. Let $d$ be the smallest distance between all pairs of points constructed (including points on the polygon). The points constructed by the algorithm are close to these $O(n^4)$ points. In particular, for every point computed by the algorithm, one of the $O(n^4)$ points is within a distance of $h \epsilon$ from it (each new point is at a distance of at most $\epsilon$ away from a previously known point by construction).
Thus, to ensure that the points computed are reasonably correct, it is sufficient to choose $\epsilon$ such that $h \epsilon < \frac{d}{2}$. Suppose that all the points can be represented accurately with $c \log n$ bits (i.e. they are on an integer grid of size $[0, n^c] \times [0, n^c]$), then $d \leq \frac{1}{n^c}$ and consequently $\epsilon < \frac{1}{2h n^c}$ is sufficient. This shows that if we can compute the points by rounding upto $\log{2h} + c\log{n}$ additional bits, then the set of points computed by the algorithm remains a valid solution to the cooperative guards problem.

%% file: model.tex

\section{Model}\label{model}

Our models are closely related to the existing finite communication and finite state models which have been established for robot exploration on general graphs.
Refer \cite{DCME_book, FLOCCHINI201657, sugihara1996distributed} for a thorough survey of such models. 
In our case the agents move inside a polygonal region with $n$ vertices and $h$ holes. 

We assume that initially a minimum of $\frac{n+2h-2}{2}$ agents are deployed at a single point in the polygon $P$. 
Each agent is modeled as a deterministic processor with some finite (not necessarily {\bf constant}) amount of persistent memory that can be used for local computations. The agents are considered as point objects on the $2D$ plane.
Additionally we also assume that each agent has a set of externally visible lights that it can choose to turn on or off. These colors allow the agents to broadcast bits to all agents that are visible to it. Each agent has a unique ID associated with it. 
All the agents operate in synchronised $\texttt{Look-Compute-Move}$ (LCM) cycles. 

\textbf{Look}. Each agent orients itself in a particular direction and performs a $360$ degree sweep of its visibility polygon. After the look operation, the agent has knowledge of its visibility polygon, the vertices of the visibility polygon that are polygonal vertices and all agents that are on or within the visibility polygon as co-ordinates with respect to a common reference frame. An agent can also see the colors that are turned on in all the agents that are visible to it.

\textbf{Compute}. Each agent executes an algorithm (which is same for all agents) whose result is a destination point within its visibility polygon (which could possibly be the same as the current location).

\textbf{Move}. Each agent moves to the destination point computed in the previous compute cycle and moves to the destination in a straight line.

\textbf{Communication}. In order to enable communication, we provide each agent with a set of externally visible lights that can be turned on or off in the $\texttt{compute}$ cycle. In order to send bits to other agents, an agent $A$ can compute and set the corresponding lights in the current $\texttt{compute}$ cycle and in the subsequent $\texttt{look}$ cycle, any agent that can see $A$ can deduce the set of bits transmitted by $A$. Thus each communication cycle can be viewed as a local broadcast. 

\textbf{Depth Perception.} We call the model described above, the {\em depth perception} model, owing to the fact that the agents can perceive distances clearly and determine co-ordinate of other vertices. We assume that in this case the vertices of the polygon are integer points within an $[0, n^c] \times [0, n^c]$ integer grid for some constant $c$. As a consequence $O(\log n)$ bits are sufficient to compute the co-ordinates. Moreover, all points computed by our algorithms have rational co-ordinates and can be represented precisely within $O(\log n)$ bits.

\textbf{Proximity Perception.}
Note that the $\texttt{Look}$ cycle is the only cycle that captures information about the polygon. 
In this case, the agents are powerful enough to completely and accurately map their visibility polygon.
In a realistic scenario, it might not be possible to obtain such parameters with a good enough precision.
An algorithm that does not rely on the computation of such co-ordinates, would be more favourable (and also more {\it elegant}).
It would also enable the use of less accurate mapping schemes (such as Photogrammetry) during a \texttt{Look} cycle.
It is thus of both theoretical and practical interest to us to determine how much information about the polygon is required from a $\texttt{Look}$ cycle in order to solve the {\it Cooperative Guards} problem, which leads us to a consider another model, which we call {\em proximity perception} in which agents have restricted sensing abilities. 
This differs from the {\em depth perception} model only in the manner in which the visibility polygon and positions of other agents  obtained by each agent during a {\texttt{look}} cycle are encoded. 
Under proximity perception, the output of a \texttt{Look} cycle is a pair of sequences, one representing the {\em vertex-limited } visibility polygon (See Definition \ref{def:visbility_vertices}) and the other representing the agents that are within its visibility polygon. 
Both sequences are sorted, first by the angle that they subtend (with respect to the reference frame of the agent) and then by the distance from the agent. 
Figure \ref{fig:look_cycle} shows an example of the output of a {\texttt {look}} cycle. In this model, we only receive sequences of vertices and agents but not their co-ordinates. 
Thus, the view of each agent is a combinatorial object that can be determined without the ability to perceive {\em exact} distances.
This combinatorial object enables the agents to compare the angles subtended by two points in its visibility polygon. 

\textbf{Parameters of an Algorithm}. We quantify or measure an algorithm by several parameters. The first is the round complexity. The second is total number of broadcasts performed by all agents. The third, is the amount of {\it persistent} memory that each agent has. Persistent memory are those set of bits that are never lost and are stored throughout the algorithm. In order to map out the visibility polygon, each agent potentially needs a large amount of bits since the visibility polygon could have $\Theta(n)$ vertices, however such information does not need to be stored across multiple rounds. We assume that the agents have sufficient number of volatile bits that are reliable for only one round, and these volatile bits are used to store the visibility polygon and perform additional computation as required.





\begin{figure}
\includegraphics[width=0.8\textwidth]{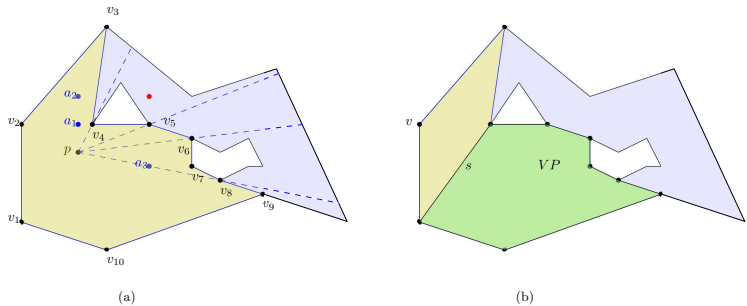}
\caption{(a)The vertex limited visibility polygon obtained for the point (agent) $p$ is shaded yellow. The two sequences returned during a look cycle an agent at $p$ are $\{v_1, v_2, \dots v_{10}\}$ and $\{a_1, a_2, a_3\}$. The red point is an agent not visible to $p$. Note that $a_3$ does not obstruct the visibility of $v_8$ from $p$. (b) crop($VP, s, v$) (shaded in green)}\label{fig:look_cycle}
\end{figure}


%% file: algorithm.tex

\section{Depth Perception Model}\label{sec:algorithm}

In this section we provide algorithms for deploying agents that solve the Cooperative Guards problem in polygons with holes for the {\em depth perception} model. We provide algorithms that trade-off between the amount of persistent memory required by an agent and the number of broadcasts performed. The round complexity of all the algorithms are linear in $n$.

\textbf{Warmup: Agents with large memory capacities.} If one agent has sufficient amount of memory, then we can map out the entire polygon by an exploration algorithm which runs in $O(n)$ rounds. We first choose the agent that is responsible for mapping the polygon.
 (The agents have unique IDs so for example the agent with the least ID maybe chosen as a ``leader'').
Since all information is stored in one agent, the agent can start at any arbitrary vertex, compute the vertex-limited visibility polygon and then move to all the new vertices obtained, compute their vertex limited polygons and so on. By merging all the visibility polygons we can obtain the co-ordinates of all vertices of the polygon and the sides. The movement can be visualised as a graph traversal on a spanning tree of the vertices of the polygon and hence runs in $O(n)$ rounds. Once the polygon is obtained, the agent can compute the guard set as described in Section \ref{warmup} (with the reduction to $\frac{n+h-2}{2}$ agents) and then assign triplets to the other agents (perhaps in increasing order of their IDs). This gives our first result in a straightforward manner. Refer Algorithm \ref{alg:dist:warmup} for the pseudo-code.

\begin{algorithm}
    \caption{Warmup Algorithm}
    \label{alg:dist:warmup}
    \begin{algorithmic}[1]
        \State \Comment{Code executed by each node $v$}
        \State \textbf{Phase 1:} Leader Election \Comment{Takes $1$ round}
        \State \textbf{Broadcast: } ID \Comment{Each node broadcasts its own ID}
        \If{ID is least among all visible IDs}
            \State $\mathsf{leader} \gets \mathsf{true}$ 
            \State Push current location to a stack $S$
        \EndIf
        \State \textbf{Phase 2:} Graph Exploration \Comment{Takes $O(n)$ rounds}
        \If{$\mathsf{leader}$}
            \State \textbf{Look : } Obtain the vertex limited visibility polygon $VP$
            \State \textbf{Compute : }
            \State Merge $VP$ to the locally computed map of the polygon $P$
            \ForAll{vertices $v \in VP$ that have not yet been visited}
                \State Add $v$ to the list of visited vertices
                \State Push $v$ into a stack $S$
            \EndFor
            \If{$S$ is empty}
                \State Using the centralized algorithm, compute a cooperative guard set
                \State Compute an Euler Tour of any tree that spans the vertices corresponding to the cooperative guard set
                \State Trigger computation of Phase $3$
            \Else 
                \State $\mathsf{dest} \gets S.top$ \Comment{Set destination as the point on the top of the stack} 
                \State Pop top element from $S$ 
            \EndIf
            \State \textbf{Move : } Move to $\mathsf{dest}$
        \EndIf
        \State \textbf{Phase 3:} Guard placement \Comment{Takes $O(n)$ rounds}
        \If{$\mathsf{leader}$}
            \State $\mathsf{dest} \gets $ Next vertex in the Euler Tour
            \If{current position has been visited for the first time}
                \State \textbf{Broadcast:} Highest ID node that has not yet been visited and instruct it to not move
            \EndIf
            \State \textbf{Broadcast:} Co-ordinates of $\mathsf{dest}$
        \Else
            \State \textbf{Receive:} Co-ordinates of $\mathsf{dest}$ and/or instruction to not move
            \If{Instruction to not move is received}
                \State Terminate
            \EndIf 
        \EndIf
        \State \textbf{Move :} Move to $\mathsf{dest}$
    \end{algorithmic}
\end{algorithm}

\begin{figure}
\includegraphics[width=0.8\textwidth]{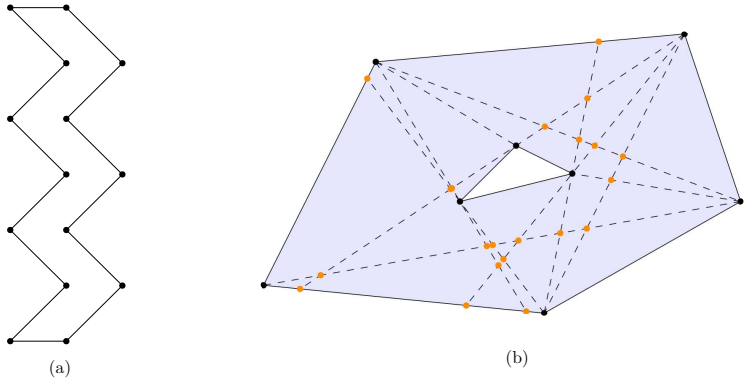}
\caption{(a) A polygon such that any set of cooperative guards have visibility graph with diameter $\Omega(n)$ (b) The decomposition into convex regions by drawing lines between all pairs of vertices} \label{fig:sec4fig}
\end{figure}

\begin{theorem}\label{thm:dist:alg1}
    There exists a distributed algorithm, in the depth perception model, that deploys $\frac{n + h - 2}{2}$ agents in a polygon with $n$ vertices and $h$ holes that takes $\Theta(n)$ rounds and performs $\Theta(n)$ broadcasts, where exactly one agent has $O(n \log n)$ bits of persistent memory and rest have $O(\log{n})$ bits.
\end{theorem}

We note that the round complexity of the algorithm obtained above is {\em existentially} optimal, in particular for every integer $n \geq 4$, there exists a polygon with $n$ vertices for which any algorithm (even if randomized) requires $\Omega(n)$ rounds. To see this, construct polygons where the diameter of the visibility graph of any co-operative guard set is $\Omega(n)$. One such example is given in Figure \ref{fig:sec4fig}.

\textbf{Agents with small memory.}\label{alg:sm} The problem is more challenging when the maximum amount of persistent memory is limited. Suppose each agent has at most $O(\log{n})$ bits of persistent memory, then we can still attain the same round complexity at the cost of performing more number of broadcasts. 

{\em Phase 1. } Setup and Intuition. Since a single agent cannot remember the entire polygon, we split the information equally among all the  $O(n)$ agents. We simultaneously explore and build a triangulation of the polygon. The exploration is  different from the one described above, since we need to construct the triangulation as well, we perform a graph traversal on the dual graph $\mathcal{D}_T$ of the polygon. We do not know the dual graph in advance, but we will construct it as we go along. Before starting the exploration we order the agents by their IDs and compute, for each agent, the agent with the next smallest and next largest IDs. We can perform these computations in a single round with $O(n)$ broadcasts. The ordering will help us to choose the agents to store some given information once a new triangle is computed. In this algorithm, during the exploration phase, we have all agents move together simultaneously. 

{\em Phase 2. } \label{alg:sm:2} Polygon Exploration. We construct the first triangle such that it has degree at most $2$ in the dual graph (as in Section \ref{warmup}). The computed triangle and its coordinates are stored by one of the agents. Once a triangle, say $\tau$, has been constructed, we look to construct triangles adjacent to the sides of $\tau$. To determine such triangles, we move all agents to one of the ends of the sides (say $s$) of $\tau$ that needs to be explored. Then compute the vertex limited visibility polygon $VP$ and crop out the portion of $P$ on the same side of $s$ as the triangle $\tau$ (see Definition \ref{def:crop}), since this part of the polygon is already covered. Let the cropped polygon be $P$.

Next, we seek a point $v \in P$ such that the triangle formed by the ends of $s$ and $v$ does not intersect any other triangle placed so far. For some triangle $\tau'$, the set of such vertices belonging to $P$ can be described as $O(1)$ contiguous sub-sequences of the vertices of the polygon $P$ and this information can be broadcast in a single round since it requires only $O(\log n)$ bits. For all triangles computed so far, these triangles are stored in some agent. Each of these agents  can then compute the corresponding $O(1)$ sub-sequences for the $O(1)$ triangles that they have been assigned to and then broadcasts this information. The leader can then take the intersection of all the sub-sequences and compute the next valid triangle. Thus, constructing a new triangle requires $O(1)$ rounds and $O(n)$ broadcasts. We can repeat this procedure to construct all the $n + 2h - 2 = O(n)$ triangles. While adding the new triangles we can also store parent pointers to (the agents containing) the previous triangle, so the dual graph can also be constructed.

{\em Phase 3. } \label{alg:sm:3} Placing Guards. Once the triangulation has been computed, we can simulate the centralized Algorithm \ref{alg:dfs} in $O(n)$ rounds and $O(n)$ broadcasts. Pseudo code for each phase is presented in \ref{code:sm},\ref{code:sm2} and \ref{code:sm3}. This leads to our next result. 

\begin{algorithm}
    \caption{Phase 1: Algorithm for agents with small persistent memory}
    \label{code:sm}
    \begin{algorithmic}[1]
        \Statex {\bf {Phase 1} : $O(1)$ rounds} \Comment{Order the agents, and construct the first triangle}
        \State {\bf Broadcast.} ID 
        \State Compute $rank, pred, succ$ and store them
        \State Move to any vertex $v$ (chosen by leader) of the polygon
        \State If leader, compute a triangle $t$ that has a side $s$ of the polygon and vertex $v$ 
        \Statex \Comment{Computed and stored only by the leader, i.e. agent with $rank=1$}
        \Statex \Comment{$t$ has degree at most $2$ in the dual graph}
        \State $\texttt{p}(t) \gets t$ \Comment{parent of root is itself in dual graph}
        \State If $rank = 1$, $T \gets \{t\}$ and $S \gets $ set containing sides of $t$ \Comment{store triangles and sides}
        \State If $rank = 1$, {\bf{Broadcast.} $t$}
    \end{algorithmic}
\end{algorithm}
\begin{algorithm}
    \caption{Phase 2 : Algorithm for agents with small persistent memory}
    \label{code:sm2}
    \begin{algorithmic}[1]
        \Statex {\bf {Phase 2} : $O(n)$ rounds} \Comment{Triangulation via DFS}
        \Statex \Comment{$t \gets $ latest triangle that was broadcasted}
        \If{$t \in T$ and {\bf only} $t$ was received in previous round}
            \If{$t$ contains a side $s$ that was not previously chosen}
                \State $uv \gets $ ends of the side $s$
                \State {\bf Broadcast.} $uv$ with an instruction to move
            \Else 
                \If{$\texttt{p}(t) \neq t$}
                    \State {\bf Broadcast.} $\texttt{p}(t)$
                \Else 
                    \State {\bf Broadcast.} Move to Phase 3
                \EndIf
            \EndIf
        \EndIf
        \If {vertices $uv$ is received with instruction to move} 
            \Statex \Comment{Assume that every agent broadcasts to itself as well}
            \State Move to $u$
            \State $VP \gets $ current vertex limited visibility polygon 
            \State $S \gets$ set of vertices $w$ of $VP$ such that $\mathsf{area}(\Delta uvw \cap t) = 0$
            \State {\bf Broadcast : } $S$ as a set of $O(1)$ contiguous intervals
        \EndIf
        \If{$t \in T$ and sets $S_i$ received from all agents}
            \State $S' \gets \cap_{i} S_i$ \Comment{Compute intersection of all sets}
            \If{$S'$ is empty}
                \State {\bf Broadcast.} $t$
            \Else
                \State $w \gets $ some vertex in $S'$
                \State $t' \gets \Delta uvw$
                \If{$|T| < 5$}
                    \State Append $t'$ to $T$
                    \State $\texttt{p}(t') \gets t$
                    \State {\bf Broadcast.} $t'$
                \Else 
                    \State {\bf Broadcast.} $t', t$ with instruction to add to $succ$
                \EndIf
            \EndIf
        \EndIf
        \If{$t', t$ is received with instruction to add to current ID} 
            \State Append $t'$ to $T$
            \State Set $\texttt{p}(t')$ to $t$
            \State {\bf Broadcast.} $t'$
        \EndIf
    \end{algorithmic}
\end{algorithm}

\begin{algorithm}
    \caption{Phase 3 : Algorithm for agents with small persistent memory}
    \label{code:sm3}
    \begin{algorithmic}[1] 
        \Statex {\bf {Phase 3(a)} : $O(n)$ rounds} \Comment{Computing Co-operative guard set}
        \Statex \Comment{Initially, $t$ is the root of the dual graph (first triangle constructed)}
        \Statex \Comment{Every triangle is marked $\mathsf{uncovered}$ and $\mathsf{unvisited}$}
        \If{$t \in T$ and $t$ is the last triangle that is broadcasted} 
            \State $C \gets$ set of children of $t$ \Comment{Can be found in two rounds}
            \If{any child $c \in C$ is $\mathsf{unvisited}$}
                \State {\bf Broadcast.} $c$
            \Else
                \State $C' \gets$ set of children of $t$ that are $\mathsf{uncovered}$
                \State $T' \gets $ top $3$ triangles with highest depth in $C' \cup \{t, \texttt{p}(t), \texttt{p}(\texttt{p}(t))\}$
                \State $v \gets$ common vertex in all triangles of $T'$
                \State Broadcast and store $v$ in the agent that contained highest depth triangle in $T$
                \State If $\texttt{p}(t) \in T$, then mark $t$ as $\mathsf{covered}$
                \State If $\texttt{p}(t) = t$ then, {\bf Broadcast.} Move to Phase 3(b)
                \State Else, {\bf Broadcast.} $\texttt{p}(t)$
            \EndIf
        \EndIf
    \end{algorithmic}
\end{algorithm}

\begin{theorem}\label{thm:dist:alg2}
    There exists an algorithm, in the depth perception model, that deploys $\frac{n + 2h - 2}{2}$ agents as cooperative guards, such that
    \begin{enumerate}
        \item it runs in $O(n)$ rounds 
        \item it performs $O(n^2)$ broadcasts, each of size $O(\log n)$ bits 
        \item each agent requires at most $O(\log{n})$ bits of persistent memory.
    \end{enumerate}
\end{theorem}

\textbf{Improving Broadcast Complexity.} \label{alg:broadcast}
In the first two algorithms, while broadcasting vertex positions, since all agents were in the same position we can encode the vertex $v$ using a number from $1$ to $n$ by ordering the vertices based on the angle that they subtend w.r.t the current vertex, so we can avoid sending co-ordinates. The advantage of not having to send exact co-ordinates is that the algorithms work even if the agents are located in an arbitrary grid as long as the co-ordinates can be computed with sufficient precision. If we assume that the vertices of the polygon are in $[0, n^c] \times [0, n^c]$ integer grid, then we can send co-ordinates as well in one round (this requires only $c \log n = O(\log{n})$ bits), and consequently we can simulate parallel algorithms for triangulation that have been established in the CREW PRAM model. 

First we compute co-ordinates of the vertices of the polygon by doing a depth-first search on the polygon similar to the previous case. However instead of building the triangulation simultaneously, only perform depth-first search on the vertices of the polygon (this can be done by replacing triangles in Phase 2(\ref{alg:sm:2}) of the previous algorithm by vertices of the polygon). In the previous case, the broadcast complexity was dominated by the process of computing valid triangles by each agent, which is no longer required for us in this scenario. Thus the broadcast complexity of the depth-first search procedure is $O(n)$.

\begin{lemma}
    Any algorithm $\mathcal{A}$ that solves a problem in the CREW PRAM model with $O(n)$ processors and has a running time $T(n)$ and uses $O(n \log n)$ bits of total memory may be simulated by a set of $\Omega(n)$ agents, each capable to possessing $O(\log n)$ persistent bits of memory within $T(n)$ rounds and using $O(n T(n))$ broadcasts.
\end{lemma}

\begin{proof}
    Each agent acts as a processor that is capable of performing fundamental bitwise (and hence also arithmetic) operations using $O(\log n)$ bits. 
    
    First we split the total memory accessed by $\mathcal{A}$ equally amongst all the agents. Since the total memory used is $O(n \log n)$ and there are $\Omega(n)$ agents, each agent has at most $O(\log n)$ bits. Each agent is given the responsibility to simulate some $O(1)$ processors. 
    
    {\em Concurrent Read.} Before a read, all agents broadcast the bits in their memory locations along with their IDs. If a processor $p$ in $\mathcal{A}$ reads a memory location at some address $A$, the agent responsible for $p$ looks at the information broadcast by agent containing $A$. For all the agents, this costs $n$ broadcasts and $O(1)$ rounds.
    
    {\em Exclusive Write.} To simulate writing of a memory location by a processor $p$, the agent responsible for $p$ broadcasts this information along with the address. In the next round the agent containing the address updates the value stored in its memory. Since each agent is responsible for $O(1)$ processors, all the writes corresponding to a single round in $\mathcal{A}$ can be simulated within $O(1)$ rounds and $O(n)$ broadcasts.
    
    Each read and write costs at most $O(n)$ broadcasts and hence if the running time of the algorithm is $T(n)$, then the number of broadcasts is $O(nT(n))$.
\end{proof}

Thus, we can simulate Goodrich's algorithm \cite{GOODRICH1989327} (which runs in time $O(\log n)$ and requires $O(n \log n)$ bits) for computing the triangulation in the CREW-PRAM model which requires only $O(\log n)$ additional rounds while requiring only $O(n \log n)$ broadcasts. This gives our next result,

\begin{theorem}\label{thm:dist:alg3}
    There exists an algorithm, in the depth perception model, that deploys $\frac{n + 2h - 2}{2}$ agents, such that
    \begin{enumerate}
        \item it runs in $O(n)$ rounds 
        \item it performs $O(n \log n)$ broadcasts, each of size $O(\log n)$ bits 
        \item each agent requires at most $O(\log{n})$ bits of persistent memory.
    \end{enumerate}
\end{theorem}

{\textbf{Reducing Guards.}} We can simulate the centralized algorithm that performs the reduction to $\frac{n + h - 2}{2}$ agents by broadcasting co-ordinates for all the agents. For creating each channel, in one round the agents can broadcast all the co-ordinates of the points that it sees. Each agent can then locally maps out the entire polygon. We have each agent compute a candidate channel and communicate this information to the other agents. If multiple channels can be formed, then we select the valid channel by the least ID agent. The construction of the holes needs to be done sequentially to prevent the channels from interfering with each other. Thus, this requires $O(n)$ broadcasts for the reduction to each hole and so we can compute the reduced guard set in an additional $O(h)$ rounds and $O(n h)$ broadcasts in total. Once we have reduced the problem to a simple polygon with $n + h$ vertices, we may use any of the earlier algorithms to compute the co-operative guard set.

\begin{theorem}\label{thm:dist:alg4}
    There exists an algorithm, in the depth perception model, that deploys $\frac{n + h - 2}{2}$ agents, such that
    \begin{enumerate}
        \item it runs in $O(n)$ rounds 
        \item it performs $O(n (h + \log n))$ broadcasts, each of size $O(\log{n})$ bits 
        \item each agent requires at most $O(\log{n})$ bits of persistent memory.
    \end{enumerate}
\end{theorem}

\section{Proximity Perception Model}\label{sec:algorithm:proximity}

In the previous section, all algorithms depended on the ability to compute exact distances and co-ordinates. In this section we look at the {\em proximity perception} model, where the agents have limited sensing capabilities (refer Section \ref{model} for more details). Since co-ordinates cannot be determined, it is difficult to perform even fundamental computational geometry operations such as computing the intersection of diagonals. In this section we show that it is possible to construct a cooperative guard set using $\frac{n + 2h - 2}{2}$ agents, however our algorithm requires $O(n^4)$ rounds. 

We proceed in a manner similar to Algorithm \ref{alg:sm} described for the depth perception model. This time, instead of having all the agents move together, we place the agents on the vertices of the polygon that they store and only move the remaining agents. This will help us to perform some geometric operations such as determining whether two diagonals intersect. To explore the polygon using only $\frac{n + 2h - 2}{2}$ agents, we cannot afford to place agents on all the vertices, so our strategy is to compute, simultaneously along with the triangulation, the cooperative guard set (that guards the triangles computed so far) using the update algorithms presented in Section \ref{warmup}. We shall present the algorithm in several phases, that run concurrently with each other in regular intervals.

\subsection{Triangulation}

\textbf{ \ \ Exploration}. Every agent is either an $\texttt{explorer}$ or a $\texttt{guard}$ at any given point in time. The $\texttt{explorer}$ agents move along the polygon to discover new vertices and construct triangles. The $\texttt{guard}$ agents are placed as solutions to the Co-operative guards problem for the sub-polygon that has been explored so far. Initially all agents are $\texttt{explorer}$ agents and during the course of a round all $\texttt{explorer}$ agents move together. Once a guard position has been determined (which corresponds to a triplet of triangles), one of the $\texttt{explorer}$ agent changes its type to $\texttt{guard}$ and guards all the triangles in its triplet. The first triangle, say $\tau_0$, is computed such that its degree in the dual graph is at most $2$. 

The triangulation and exploration procedure is similar to Phase 2 (refer \ref{alg:sm:2}) of Algorithm \ref{alg:sm}, except that in this case we can no longer rely on co-ordinates and existing $\texttt{guards}$ and all agents need not be in the same location. Suppose we have constructed a set of triangles and looking to extend the triangulation by forming a triangle with side $s$ (with endpoints $uv$) neighboring to an already constructed triangle $\tau'$. The explorer agent moves to one end of $s$, constructs the cropped visibility polygon $P$ (as in \ref{alg:sm:2}). To determine if $\Delta uvw$ is valid (i.e. can be added to the triangulation) for some $w \in P$, we shall find if the line segments $uw$ and $vw$ properly intersect any of the sides of triangles constructed so far. This is the most challenging and time consuming step and is described separately in the next paragraph.

If both the diagonals $uw$ and $vw$ are valid and the triangle $uvw$ has not been constructed so far, we add $\Delta uvw$ to the triangulation and update the guard formation by emulating the update procedure (Algorithm \ref{alg:central_update}) in the distributed setting. A centralized version of the pseudo-code is provided in Algorithm \ref{alg:proximity}. We leave it to the reader to verify that each step maybe emulated in the distributed setting.

\begin{algorithm}
    \caption{Pseudo Code for the Proximity Perception Algorithm (Centralized version)}
    \label{alg:proximity}
    \begin{algorithmic}[1]
        
        \State $t \gets$ triangle whose one side is the edge of a hole or the outer boundary of the polygon
        \State Initialize dual graph with triangle $t$, and place a guard at any vertex of $t$
        \While{\texttt{true}}
            \State \Call{Explore}{ }
        \EndWhile
        
        \Procedure{Explore}{ }
            \Statex DFS based search on the polygon for new vertices and triangles
                \State $t \gets $ current triangle in the DFS exploration
                \State Find a side $s$ of $t$ that is yet to be explored
                \If{such a side $s$ does not exist}
                    \State $t' \gets$ the parent of $t$ in the dual graph 
                    \State If $t$ is the root, then terminate, otherwise move to $t'$
                \Else
                    \State $VP \gets $ vertex limited visibility polygon at current vertex
                    \State $w \gets $ vertex of $t$ that is not on $s$
                    \State $u, v \gets$ vertices on ends of $s$
                    \State $S \gets \texttt{crop}(VP, s, w)$ 
                    \ForAll{vertices $q$ in $S$}
                        \If {\Call{Validate}{$u, q$} $\And$ \Call{Validate}{$v, q$}}
                            \State Add the triangle $t' \gets \{u, v, q\}$
                            \State \Call{Update}{$t', t$} \Comment{Algorithm \ref{alg:central_update}}
                        \EndIf
                    \EndFor
                \EndIf
        \EndProcedure
        \Procedure{Validate}{$u, v$}
            \Statex Determine if line joining $u, v$ intersects any other diagonal
            \State $p \gets u$ \Comment{current position}
            \While{$p \neq v$}
                \State $q \gets $ the point on line $pv$ closest to $u$ such that visibility polygon changes 
                \Statex \Comment{From Lemma \ref{lemma:ghosh}, point $q$ must be at the intersection of two diagonals}
                \If{there exists a diagonal intersecting $pv$}
                    \State \Return \texttt{false}
                \EndIf
            \EndWhile
            \State \Return \texttt{true}
        \EndProcedure
    \end{algorithmic}
\end{algorithm}

\textbf{Validation of triangles}. Ghosh \cite{GHOSH2010718} provided an interesting relation between the intersection points of the lines joining pair of points of $P$ and the visibility polygons of points of a polygon $P$. The diagonals of $P$ and all the intersection points of all pairs of diagonals, decompose $P$ into a set of convex regions, say $\mathcal{D}$. An example of such a decomposition is shown in Figure \ref{fig:sec4fig}. The following lemma (taken from Ghosh \cite{GHOSH2010718}) shows us that this decomposition is very useful,

\begin{lemma}\label{lemma:ghosh}
    For every point $p \in P$ and every convex region $C \in \mathcal{D}$, either all points in $C$ are visible from $p$, or none of them are visible from $p$.
\end{lemma}

We can leverage the above decomposition to determine if two diagonals intersect. To check if a diagonal between two vertices $u_1, u_2$ and a diagonal between vertices $v_1, v_2$ intersect, we place an agent at $u_1$ (chosen arbitrarily) and then move on the line joining $u_1, u_2$ until the vertex-limited visibility polygon sees different vertices than those seen at $u_1$. Repeatedly constructing those points on the line, we construct a subset of the intersection points computed above. If $u_1, u_2$ and $v_1, v_2$ intersect, either the agent is exactly at the intersection point or there are two points on either side of the diagonal $v_1v_2$ such that the agent is present at the points in consecutive rounds. In both cases an agent at $v_1$ can determine by looking at the order of the angles, if the agent moving from $u_1$ to $u_2$ crosses the diagonal and communicate the results.  The number of such points in each diagonal is $O(n^2)$ and the agent moves to each in constant number of rounds (See proximity perception model in Section \ref{model}). Since there are in total at most $O(n)$ candidate diagonals to be checked to extend each of the $O(n)$ edges of the triangulation, visiting all such points throughout the algorithm requires $O(n^4)$ rounds. The above procedure works only if an agent is present at $v_1$, which is not always the case, since there are at most $\frac{n + 2h - 2}{2}$ agents placed and $n$ vertices. However note that each vertex $v_1$ is part of a triplet and the agent corresponding to the triplets store the $O(1)$ vertices corresponding to its visibility polygon. Thus, we can have the $\texttt{guard}$ agents patrol all the vertices that it stores, for the sole purpose of facilitating the detection of intersection of diagonals. The patrol takes exactly $5$ rounds and we can ask the agent moving from $u_1$ to $u_2$ to wait for $5 = O(1)$ rounds to ensure that all diagonals have been checked. 

\textbf{Update Guard Placement.} Once a candidate triangle is found, we can simulate the triplet update algorithm presented in Section \ref{warmup} to update the guard set. Each guard stores the triangles corresponding to its triplet, along with the vertices and edges. Since there are at most $5$ vertices and $7$ sides corresponding to each triplet, we can store this information using $O(\log n)$ bits (as persistent memory). The update algorithm for a node $v$ of the dual graph is then executed by the guard corresponding to the triplet containing it. We keep track of the number of triangles constructed so far, to determine if a new guard is required. When the number of triangles increases to an even number that is at least $4$, the number of triplets (and agents) increases by one. At the location where the new triangle is created, either a new agent is deployed for the triplet containing the new node or the existing triplet is modified. The triplets already placed covering the parent of the node are pushed above and eventually the entire triplet configuration is constructed. The triplet update can be run in parallel with the exploration procedure above, however the run time is dominated by the time required to determine if diagonals intersect. This gives our final result,

\begin{theorem}\label{thm:dist:proximity}
    There exists an algorithm, in the proximity perception model, that places $\frac{n + 2h - 2}{2}$ cooperative guards in $O(n^4)$ rounds. 
\end{theorem}

The reduction to $\frac{n+h-2}{2}$ guards does not seem easy. The previous methods required the computation of co-ordinates of intersection points of lines inside the polygon, which is not possible in this model.

%% file: conclusion.tex

\section{Conclusion and Future Works}

In this paper, we have studied the Cooperative Guards from a distributed mobile agent perspective. We studied two different models where agents had different sensing capabilities. In one model, which we call the {\it depth perception} model, the agents are able to entirely map out the polygon by estimating co-ordinates. The second model, which we call the {\it proximity perception} model, the agents have limited sensing capabilities. We compare the difficulties of solving the problem in the two models, providing efficient linear time algorithms in the stronger model and a polynomial ($O(n^4)$) time algorithm in the weaker model. In the stronger model, we provide multiple algorithms that run in linear number of rounds, that trade-off between the number of broadcasts performed, the amount of information sent in a single round and the number of guards placed.
We discuss about further improvements and some closely related open problems below.

\textbf{Open Problem 1. } {\em {How many co-operative guards are required for a polygon with $n$ vertices and $h$ holes?  What if guards are restricted to be on the vertices of the polygon? }} 
In Section \ref{warmup}, we gave a construction that required $\frac{n + h - 2}{2}$ co-operative guards, however we have not been able to construct polygons requiring so many guards. For $h \leq 1$, the problem has been solved by Zylinski \cite{zylinski2005cooperative}, but for $h \geq 2$, the problem remains open. We have not been able to provide a better upper bound than $\frac{n + 2h - 2}{2}$ guards when the guards are restricted to the vertices of the polygon. We have also not been able to construct polygons where the number of {\em vertex guards} required is more than $\frac{n + h - 2}{2}$, so we believe that there might exist a way to place $\frac{n + h - 2}{2}$ guards, by placing them only on the vertices of the polygon.

\textbf{Open Problem 2. } {\em {Is there a more efficient centralized algorithm to compute a cooperative guard set with $\frac{n + h - 2}{2}$ guards?}}
The centralized algorithm discussed in Section \ref{warmup} runs in $O(n^2)$ time, whereas we are able to construct a cooperative guard set using no more than $\frac{n + 2h - 2}{2}$ agents in $O(n \log n)$ time. Is it possible to bridge the gap in the running time? One possible approach might be to construct the reduction proposed by Sachs and Souvaine~\cite{BjorlingSachs1995} by using a divide and conquer algorithm on the holes.

\textbf{Open Problem 3.} {\em {Is it possible to construct $O(poly(D))$ time algorithms where $D$ is the hop diameter of the polygon, when agents have limited persistent memory? }}
The hop diameter of the polygon is the minimum number of straight line segments required to connect any two vertices of the polygon. In simple polygons, it is easy to construct such algorithms. However in polygons with holes, if the agents split up at any point for exploration, we could end up in a situation where two guards are symmetric with respect to each other. We suspect that randomization might help, though it is not clear to us, whether it is even possible when the agents are deterministic.



%% file: main.bbl
\begin{thebibliography}{10}

\bibitem{AhlfeldUweandHecker92}
Uwe Ahlfeld and Hans-Dietrich Hecker.
\newblock The computational complexity of some guard sets in polygons.
\newblock {\em J. Inf. Process. Cybern.}, 28(6):331–342, January 1992.

\bibitem{ASANO1986221}
Takao Asano, Tetsuo Asano, and Ron~Y Pinter.
\newblock Polygon triangulation: Efficiency and minimality.
\newblock {\em Journal of Algorithms}, 7(2):221--231, 1986.

\bibitem{Bharat2020}
Barath Ashok, John Augustine, Aditya Mehekare, Sridhar Ragupathi, Srikkanth
  Ramachandran, and Suman Sourav.
\newblock Guarding a polygon without losing touch.
\newblock In Andrea~Werneck Richa and Christian Scheideler, editors, {\em
  Structural Information and Communication Complexity}, pages 91--108, Cham,
  2020. Springer International Publishing.

\bibitem{BARYEHUDA1994}
Reuven Bar-Yehuda and Bernard Chazelle.
\newblock Triangulating disjoint jordan chains.
\newblock {\em International Journal of Computational Geometry {\&}
  Applications}, 04(04):475--481, December 1994.

\bibitem{BGP17}
Pritam Bhattacharya, Subir~Kumar Ghosh, and Sudebkumar Pal.
\newblock Constant approximation algorithms for guarding simple polygons using
  vertex guards, 2017.
\newblock \href {http://arxiv.org/abs/1712.05492} {\path{arXiv:1712.05492}}.

\bibitem{BjorlingSachs1995}
I.~Bjorling-Sachs and D.~L. Souvaine.
\newblock An efficient algorithm for guard placement in polygons with holes.
\newblock {\em Discrete {\&} Computational Geometry}, 13(1):77--109, January
  1995.
\newblock \href {https://doi.org/10.1007/bf02574029}
  {\path{doi:10.1007/bf02574029}}.

\bibitem{Chazelle1991}
Bernard Chazelle.
\newblock Triangulating a simple polygon in linear time.
\newblock {\em Discrete {\&} Computational Geometry}, 6(3):485--524, September
  1991.
\newblock \href {https://doi.org/10.1007/bf02574703}
  {\path{doi:10.1007/bf02574703}}.

\bibitem{CHVATAL197539}
V.~Chvátal.
\newblock A combinatorial theorem in plane geometry.
\newblock {\em Journal of Combinatorial Theory, Series B}, 18(1):39 -- 41,
  1975.

\bibitem{deshpande2007pseudopolynomial}
Ajay Deshpande, Taejung Kim, Erik~D Demaine, and Sanjay~E Sarma.
\newblock A pseudopolynomial time o (logn)-approximation algorithm for art
  gallery problems.
\newblock In {\em Workshop on Algorithms and Data Structures}, pages 163--174.
  Springer, 2007.

\bibitem{Eidenbenz2001}
S.~Eidenbenz, C.~Stamm, and P.~Widmayer.
\newblock Inapproximability results for guarding polygons and terrains.
\newblock {\em Algorithmica}, 31(1):79--113, Sep 2001.

\bibitem{FISK1978374}
Steve Fisk.
\newblock A short proof of chvátal's watchman theorem.
\newblock {\em Journal of Combinatorial Theory, Series B}, 24(3):374, 1978.

\bibitem{FLOCCHINI201657}
P.~Flocchini, N.~Santoro, G.~Viglietta, and M.~Yamashita.
\newblock Rendezvous with constant memory.
\newblock {\em Theoretical Computer Science}, 621:57--72, 2016.
\newblock URL:
  \url{https://www.sciencedirect.com/science/article/pii/S0304397516000499},
  \href {https://doi.org/https://doi.org/10.1016/j.tcs.2016.01.025}
  {\path{doi:https://doi.org/10.1016/j.tcs.2016.01.025}}.

\bibitem{DCME_book}
Paola Flocchini, Giuseppe Prencipe, and Nicola Santoro, editors.
\newblock {\em Distributed Computing by Mobile Entities, Current Research in
  Moving and Computing}, volume 11340 of {\em Lecture Notes in Computer
  Science}.
\newblock Springer, 2019.
\newblock \href {https://doi.org/10.1007/978-3-030-11072-7}
  {\path{doi:10.1007/978-3-030-11072-7}}.

\bibitem{garey1978triangulating}
Michael~R Garey, David~S Johnson, Franco~P Preparata, and Robert~E Tarjan.
\newblock Triangulating a simple polygon.
\newblock {\em Information Processing Letters}, 7(4):175--179, 1978.

\bibitem{G87}
Subir~Kumar Ghosh.
\newblock Approximation algorithms for art gallery problems.
\newblock In {\em Proceedings of Canadian Information Processing Society
  Congress}, page 429–434, 1987.

\bibitem{GHOSH2010718}
Subir~Kumar Ghosh.
\newblock Approximation algorithms for art gallery problems in polygons.
\newblock {\em Discrete Applied Mathematics}, 158(6):718--722, 2010.
\newblock URL:
  \url{https://www.sciencedirect.com/science/article/pii/S0166218X09004855},
  \href {https://doi.org/https://doi.org/10.1016/j.dam.2009.12.004}
  {\path{doi:https://doi.org/10.1016/j.dam.2009.12.004}}.

\bibitem{GOODRICH1989327}
Michael~T Goodrich.
\newblock Triangulating a polygon in parallel.
\newblock {\em Journal of Algorithms}, 10(3):327--351, 1989.
\newblock URL:
  \url{https://www.sciencedirect.com/science/article/pii/0196677489900321},
  \href {https://doi.org/https://doi.org/10.1016/0196-6774(89)90032-1}
  {\path{doi:https://doi.org/10.1016/0196-6774(89)90032-1}}.

\bibitem{hernandez1994controlling}
Gregorio Hernandez-Penalver.
\newblock Controlling guards.
\newblock In {\em CCCG}, pages 387--392, 1994.

\bibitem{1057165}
D.~{Lee} and A.~{Lin}.
\newblock Computational complexity of art gallery problems.
\newblock {\em IEEE Transactions on Information Theory}, 32(2):276--282, March
  1986.

\bibitem{Liaw1993TheMC}
B.~Liaw, N.~Huang, and Richard C.~T. Lee.
\newblock The minimum cooperative guards problem on k-spiral polygons.
\newblock In {\em CCCG}, 1993.

\bibitem{OGB11}
K.~J. Obermeyer, A.~Ganguli, and F.~Bullo.
\newblock Multi-agent deployment for visibility coverage in polygonal
  environments with holes.
\newblock {\em International Journal of Robust and Nonlinear Control},
  21(12):1467--1492, 2011.

\bibitem{Pinciu}
Val Pinciu.
\newblock A coloring algorithm for finding connected guards in art galleries.
\newblock In Cristian~S. Calude, Michael~J. Dinneen, and Vincent Vajnovszki,
  editors, {\em Discrete Mathematics and Theoretical Computer Science}, pages
  257--264, Berlin, Heidelberg, 2003. Springer Berlin Heidelberg.

\bibitem{sugihara1996distributed}
Kazuo Sugihara and Ichiro Suzuki.
\newblock Distributed algorithms for formation of geometric patterns with many
  mobile robots.
\newblock {\em Journal of robotic systems}, 13(3):127--139, 1996.

\bibitem{zylinski2005cooperative}
Pawel Zylinski.
\newblock Cooperative guards in art galleries with one hole.
\newblock {\em Balkan Journal of Geometry and Its Applications}, 10(2):142,
  2005.

\bibitem{zylinski2007cooperative}
Pawel Zylinski.
\newblock Cooperative guards in art galleries.
\newblock {\em Dissertationes Mathematicae}, pages 4--132, 01 2008.

\end{thebibliography}
